\documentclass[11pt,a4paper]{article}

\usepackage{a4wide}
\usepackage{graphicx}
\usepackage{amssymb}
\usepackage{xspace}
\usepackage{dblfloatfix}

\graphicspath{{./images/}}

\newtheorem{defin}{Definition}
 \newenvironment{definition}{\begin{defin} \sl}{\end{defin}}
\newtheorem{theo}[defin]{Theorem}
 \newenvironment{theorem}{\begin{theo} \sl}{\end{theo}}
\newtheorem{lem}[defin]{Lemma}
 \newenvironment{lemma}{\begin{lem} \sl}{\end{lem}}
\newtheorem{coro}[defin]{Corollary}
 
 \newtheorem{obs}[defin]{Observation}
 
 \newtheorem{con}[defin]{Conjecture}

\newenvironment{proof}{\emph{Proof.}}{\hfill $\Box$\\}

\newcommand{\true}{{\sc true}\xspace}
\newcommand{\false}{{\sc false}\xspace}

\title{Geometry and Generation of a New Graph Planarity Game}

\author{Rutger Kraaijer\thanks{Dept. of Information and Computing Sciences, Utrecht University, the Netherlands} 
\and Marc van Kreveld$^*$
\and Wouter Meulemans\thanks{Dept. of Mathematics and Computer Science, Eindhoven University of
Technology, the Netherlands}
\and Andr\'e van Renssen\thanks{School of Computer Science, University of Sydney, Australia}}

\date{}

\begin{document}
\maketitle

\begin{abstract}
We introduce a new abstract graph game, \textsc{Swap Planarity}, where the goal is to reach a state without
edge intersections and a move consists of swapping the locations of two vertices
connected by an edge. We analyze this puzzle game using concepts from graph theory
and graph drawing, computational geometry, and complexity. Furthermore, we specify quality criteria for puzzle instances, and describe a method to generate high-quality instances. We also report on
experiments that show how well this generation process works.
\end{abstract}

\section{Introduction}

{\sc Planarity}~\cite{Planarity} is a popular abstract puzzle game that is widely available.
Besides being a smartphone app and having a Wikipedia page, it is also
available as a ``model'' in  Netlogo~\cite{tisue2004netlogo}.
The idea is that a tangled graph is given with
intersecting edges, and the objective is to untangle the graph by dragging
vertices to other locations as to reach a plane drawing. If the graph is planar (meaning that it can be
embedded in the plane without intersections), then the objective can always be realized, and we
never need more vertex drags than there are vertices.

Algorithmically, planarity of a graph can be tested in linear time~\cite{chiba1985linear,de1990draw,tamassia2013handbook},
and the algorithm returns an embedding of the graph in which it is drawn planar.
So for an algorithm, an instance of {\sc Planarity} is easily solvable in linear time. Minimizing the number of moves, however, is NP-hard~\cite{goaoc2009untangling,Verbitsky08}, see also~\cite{BoseDHLMW09}.

In this paper we propose several variations on the game {\sc Planarity}.
These variations essentially limit the freedom of the operations that can be done
on the drawn graph. We investigate
one of the new variations closely: \textsc{Swap Planarity}, where we can \emph{swap the locations of
two vertices} that are connected by an edge. Examples are shown in Fig.~\ref{fig:introex}. 
We show that quadratically many swaps are sometimes necessary, even if the input has just one edge crossing. We also prove that, if a planar state can be reached, quadratically many swaps are always sufficient to reach it.
We also show, however, that deciding whether such a planar state exists is NP-complete for general graphs. 
Simple graphs like trees can always be made planar by swaps, but we show that minimizing the number of swaps needed is NP-complete.

\begin{figure}[bhtb]
\centering
\includegraphics{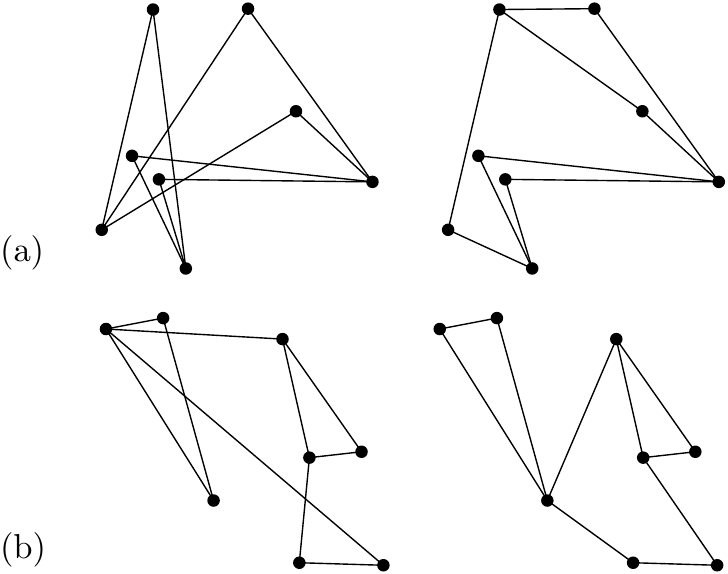}
\caption{(a) Puzzle and solution after one swap (the left, nearly vertical edge). (b) Puzzle and solution after two swaps.}
\label{fig:introex}
\end{figure}

We also investigate the automated generation of good puzzle instances.
We describe a five-step process which yields a puzzle instance.
Some of the considerations of a good instance are puzzle (complexity) based and
some are geometry based. Our process guarantees that the puzzle and
geometry criteria are met.

We implemented the generator and ran a number of experiments that uncover some properties of point set generation and puzzle diversity. The implementation includes a puzzle mode where the user can
solve generated instances by hand.

\section{Graph untangling puzzles}

We will limit the operations that change the drawing of the graph to arrive at
different puzzles. Since the puzzle type is abstract, it is necessary that
the interaction and operations themselves are simple. The puzzle then becomes
an elegant abstract puzzle of which there are many already
(Move, Lines/Flow, Zengrams, Nintaii, Fling, and several more).

Besides interacting with a vertex like in {\sc Planarity}, it is natural to
interact with an edge. Clicking or selecting is arguably the easiest interaction.
We list a number of ways in which the graph drawing can change when an edge is
selected:

\begin{description}
\item[Swap:] the two endpoints of the selected edge swap locations.
Intuitively, the edge turns around while the endpoints drag all incident edges with them.
\item[Rotate:] like swap, but now the selected edge rotates over $90$ degrees around its center. Since a single
edge can be selected consecutively three times, it does not matter whether we rotate
clockwise or counter-clockwise. 
\item[Stretch:] the selected edge is scaled by a factor $2$ from its center, or by a factor $1/2$.
\item[Collapse:] the endpoints of the selected edge are united. The united vertex is placed in the middle of the edge and gets all edges incident to the original vertices.
The selected edge is removed.
\end{description}

Of these versions, the first one distinguishes itself from the others because no
new vertex positions appear. The graph will always be drawn on the original
positions. Furthermore, the last version distinguishes itself by the fact that
the number of vertices is reduced. Eventually, the whole graph could be reduced
to a single vertex, so the challenge must be to remove all intersections in a
limited number of steps. In the first three versions, steps are reversible.

We can also stay closer to the original {\sc Planarity} puzzle and drag vertices
in more limited ways. For example, a set of points can be given along with the graph,
and the vertices must be dragged to the given points. This version is related to a well-known
problem in the graph drawing research area, namely that of embedding a graph on a
given set of points~\cite{cabello2006planar}. In essence, the initial drawing of the graph is irrelevant.

In this paper we concentrate on the swap version, named {\sc Swap Planarity}. It is perhaps the most elegant version
and the graphs appearing after operations can be controlled in their appearance, unlike
with the other versions (where edges may get so short that they cannot be selected
any more). All following results concern this version.

\begin{figure}[htb]
\centering
\includegraphics{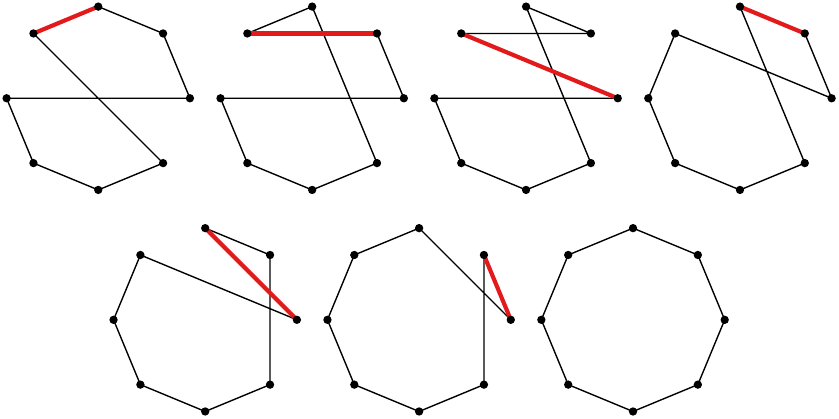}
\caption{Six steps to solve an $8$-cycle with one intersection. The edge to be
swapped is indicated.}
\label{fig:example}
\end{figure}

Before we go into the algorithmic complexity of solving such puzzles and the process
of generating good puzzle instances, we give a few examples to understand the puzzle better.
First, consider the puzzle instance in Fig.~\ref{fig:example}
with eight vertices and eight edges.
The graph is a single cycle and it has only one intersection.
To solve this puzzle, note that any swap will \emph{increase} the number of intersections.
The minimum number of swaps needed is six; the set of intermediate drawings is shown
in the figure and the selected edge is shown. When we extend this example to a set of
$n$ vertices and edges, we need $\Omega(n^2)$ swaps to solve the instance.

\begin{lemma}
There exist graphs with $n$ vertices that require $\Omega(n^2)$ swaps to obtain a plane drawing.
\end{lemma}
\begin{proof}
Consider the drawing of Fig.~\ref{fig:example} generalized to $n$ vertices, with $n$ even.
Name the vertices of the graph $v_1,\ldots,v_n$ so that $v_1,\ldots,v_{n/2}$ are
clockwise and $v_{n/2+1},\ldots,v_n$ are counter-clockwise. This implies that the
edges $(v_1,v_n)$ and $(v_{n/2},v_{n/2+1})$ intersect. Let us name the positions for
the vertices $p_1,\ldots,p_n$, where initially $v_1$ is at $p_1$ and the positions
are numbered clockwise, see Fig.~\ref{fig:lowerbound}.

\begin{figure}[htb]
\centering
\includegraphics{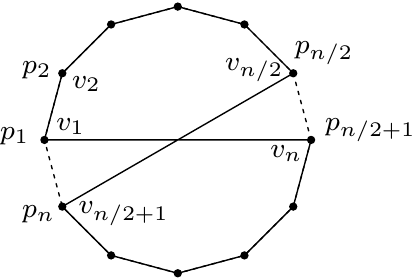}
\caption{Positions and vertices for the lower bound construction.}
\label{fig:lowerbound}
\end{figure}
In total there are $2n$ ways to place $v_1,\ldots,v_n$ on $p_1,\ldots,p_n$ without
intersections: in cyclic order clockwise or counter-clockwise, and starting anywhere.
This means that either $v_1,\ldots,v_{n/2}$ or $v_{n/2+1},\ldots,v_n$ must be reversed
on the positions $p_1,\ldots,p_n$.

Listing the points in the order of the cycle $v_1,\ldots,v_n$, we initially
get the
cyclic sequence $p_1,\ldots,p_{n/2},p_n,\ldots,p_{n/2+1}$. A swap exchanges precisely
two adjacent elements (where the first and last are also adjacent). Thus, to sort this sequence
in one of the $2n$ ways, at least ${n/2 \choose 2}= \Omega(n^2)$ swaps are needed.
\end{proof}

The next lemma shows that quadratically many swaps are sufficient; the result has been proved before as node swapping~\cite{YamanakaDIKKOSS15}.

\begin{lemma}
\label{lem:upperbound}
Every embedded graph with $n$ vertices that can reach a plane drawing using swaps has a sequence of $O(n^2)$ swaps to obtain this drawing.
\end{lemma}
\begin{proof}
Assume first that the graph has a single connected component. Name the positions
$p_1,\ldots,p_n$, and name the vertices of the graph $v_1,\ldots,v_n$ in such a way
that the graph is drawn plane if $v_i$ is at position $p_i$.
We prove by induction that any connected graph with $n$ vertices can place its
vertices at $v_1,\ldots,v_n$ at positions $p_1,\ldots,p_n$, respectively.

\begin{figure}[htb]
\centering
\includegraphics{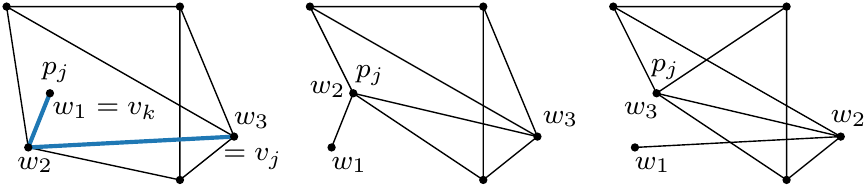}
\caption{Bringing $v_j$ to $p_j$ using the bold path $w_1,w_2,w_3$
(note that all plane embeddings must have $v_j$ at $p_j$).
Left: initial situation. Middle: after swap $(w_1,w_3)$.
Right: after swap $(w_2,w_3)$.}
\label{fig:upperbound}
\end{figure}
Choose any vertex $v_j$ such that its removal will leave the graph connected. Suppose
a vertex $v_k$ is currently at position $p_j$. Use the path between $v_j$ and $v_k$ in $G$
to get $v_j$ onto $p_j$ as follows. Suppose this path is $v_k=w_1,w_2,\ldots,w_h=v_j$,
see Fig.~\ref{fig:upperbound} for an example.
We swap $(w_1,w_2)$, then $(w_2,w_3)$, and so on until $(w_{h-1},w_h)$. This
brings $v_j$ onto $p_j$ in $h-1=O(n)$ swaps. We remove $v_j$ from the graph and $p_j$
from the locations and continue inductively. It is clear that at most $O(n^2)$ swaps are
needed in total.
If the graph has multiple connected components, we follow this procedure for each connected component.
\end{proof}

Another puzzle variant of swapping to planarity is possible, namely where
we swap any two vertices (so they need not be connected by an edge). The interaction
with the puzzle consists of clicking on two different vertices consecutively.
In this variation, any solvable puzzle instance with $n$ vertices is solvable
in at most $n-1$ swaps, because we can directly bring any vertex
to the correct position. The challenge of this variant reduces to recognizing
where vertices need to be to get a planar embedding, and no longer how to get it
there.

\section{Complexity of \textsc{Swap Planarity}}

\begin{theorem}
Given an embedded graph $G$, it is NP-complete to decide if the graph can be made planar using swaps.
\end{theorem}
\begin{proof}
A solution of the problem can be presented by the sequence in which vertices are swapped.
This solution can be represented in $O(n^2)$ space by Lemma~\ref{lem:upperbound}.
 Swapping these vertices and checking if the resulting graph is plane can be done in polynomial time, hence the problem is in NP.

Cabello~\cite{cabello2006planar} showed that it is NP-complete to decide if a given point set $P$ admits a planar drawing of a given graph $G$ where the vertices must be placed at the points. This is also true for connected graphs. Given an instance of this problem with a connected graph, we assign the vertices of $G$ to the points in $P$ arbitrarily.

We now solve the graph planarization using swaps on this embedding of the graph. If it has a solution, we can just output the final point-vertex relation, leading to a planar embedding of the given graph. If no solution exists, we also know that no planar embedding exists,
since by the proof of Lemma~\ref{lem:upperbound}, we can realize any assignment of vertices
to points in a connected graph.
\end{proof}

It is known that if $G$ is a tree, the embedding problem of $G$ onto $P$ is no longer NP-complete because every tree can be embedded without intersections onto a planar point set~\cite{bose2002embedding,pach1991embedding}. This does not imply that our puzzle game is easy to solve when the
graph is a tree \emph{when we bound the number of swaps}. In particular, we can show that deciding whether the vertices of an embedded tree can
be swapped to become plane in at most $k$ swaps is NP-complete.

\begin{theorem}
  Given an integer $k$ and a embedded tree with $n$ vertices, it is NP-complete to decide if $k$ swaps suffice to obtain a plane drawing.
\end{theorem}
\begin{proof}
  A solution of the problem can be presented by the sequence of edges to be swapped. Swapping these edges and checking if the resulting tree is plane can be done in polynomial time, hence the problem is in NP.

  For ease of explanation, the given reduction contains a number of collinear vertices. By perturbing the vertices slightly, however, the same construction works for points in general position. We reduce from positive planar 1-in-3-SAT. This problem was shown to be NP-complete by Mulzer and Rote~\cite{mulzer2008minimum}.

  \textbf{Positive planar 1-in-3-SAT.} In the positive planar 1-in-3-SAT problem we are given a collection of clauses, each consisting of exactly three variables. Each of these variables occurs positively in the clause. In addition, we are given a planar embedding of the clauses and variables such that a variable is connected to a clause if and only if the variable occurs in the clause. The positive planar 1-in-3-SAT problem asks to decide if there exists a truth assignment to the variables such that for each clause exactly one variable is true.
  
  For the reduction, we introduce gadgets for the variables, clauses and connections between these in the given embedding. We describe the construction and functioning of each gadget below, noting that an overview of the final construction is given at the end in Fig.~\ref{fig:fullconstruction}.

  \textbf{Variable gadget.} We construct a variable gadget as follows; see Fig.~\ref{fig:Variable} for an illustration. The basic construction is a path of 7 vertices, such that the first two and the last two form the corners of a square (in convex position), such that the path order matches the vertex order around the square's boundary. The remaining three vertices are placed inside the square such that their connecting edges do not intersect, but the edges $(v_2,v_3)$ and $(v_5,v_6)$ do intersect. In order to remove the created crossings using the minimum number of swaps, we need to swap the first or last edge of this path. 
  
  We now create a variable gadget, by repeating this basic construction into a longer path. The first two vertices are the last two vertices of the previous construction in the same order; effectively, we vertically mirror the basic path of every second repetition. We always use an odd number (at least 3) of repetitions. This ensures that we can remove all crossings by swapping the endpoints of all even or all odd vertical edges and this requires the same number of swaps. We designate swapping all the even vertical edges ($(v_6,v_7)$,$(v_{16},v_{17})$, etc.) to indicate \false; all odd vertical edges as \true.
  Note that this same construction can be used to propagate the truth value over longer distances as well, and that there is flexibility in this construction to make bends.

\begin{figure}[tb]
\centering
  \centering
  \includegraphics{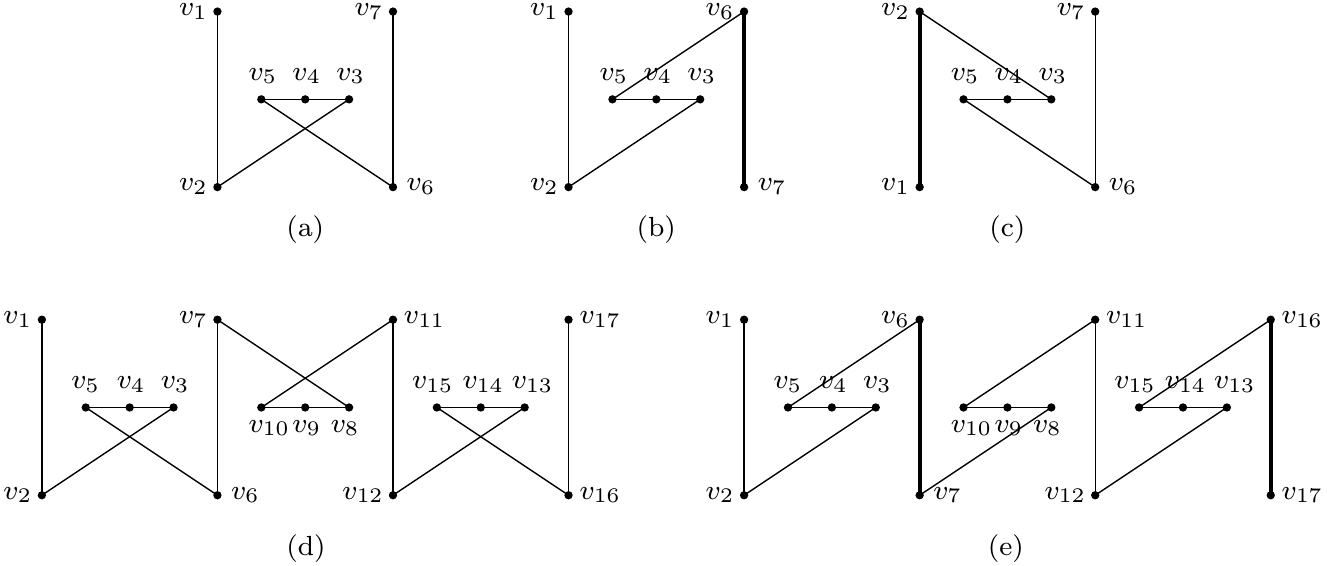}
  \caption{The variable gadget. (a) Basic construction of input path of 7 vertices. (b--c) Two ways of doing a single swap (thick edge) to untangle the basic construction. (d) Chaining the basic construction into a longer path. (e) The minimal-swap solution for a \false assignment, swapping half of the vertical edges.}
  \label{fig:Variable}
\end{figure}

  \textbf{Split gadget.} 
  In order to connect the variable to clauses, we construct a split gadget to be attached to a variable gadget; see Fig.~\ref{fig:Split}a for illustration. We take another basic construction path (indicated using the $w$-vertices) and rotate it clockwise by 90 degrees and place it below one of the even vertical edges in the variable (e.g. $(v_6,v_7)$). We shorten $(w_1,w_2)$ and lengthen $(w_6,w_7)$ such that we can place a helper vertex $h$ between $v_2$ and $v_6$ such that $h$ is inside triangle $w_1w_2v_7$ but outside triangle $w_1w_2v_6$ in the input. We then add edges $(v_6,w_1)$ and $(w_7,h)$. 
  
  There are two minimal ways of removing the crossings, each costing one swap in the split gadget, plus two in the variable construction shown in the figure. In a \false assignment, $(v_6, v_7)$ is swapped; swapping $(w_1,w_2)$ then resolves all intersections. In the \true assignment, $(v_6,v_7)$ is not swapped and swapping $(w_6,w_7)$ untangles the gadget. Observe that swapping $(w_1,w_2)$ does not resolve the intersection between $(w_1,v_6)$ and $(w_7,h)$ in the \true assignment.
  
\begin{figure}[htb]
  \centering
  \includegraphics{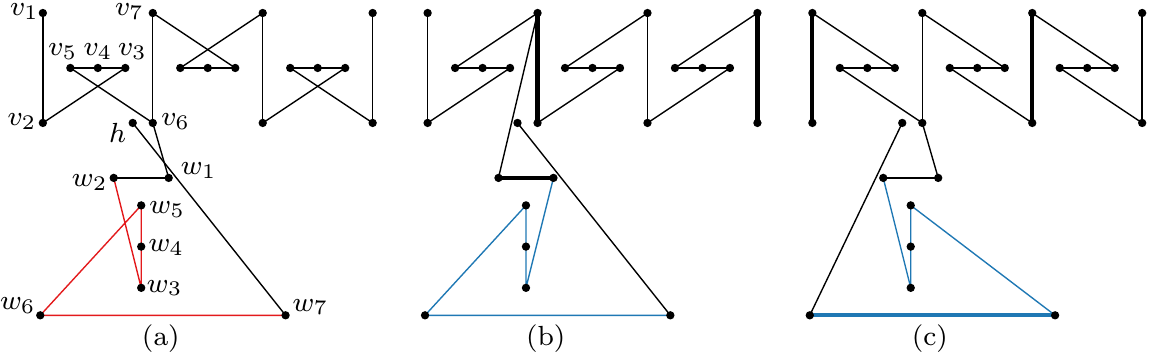}
  \caption{(a) The split gadget and its two minimal planarizations: (b) the \false assignment and (c) the \true assignment. Thick edges in (b) and (c) are those that have been swapped with respect to (a).}
  \label{fig:Split}
\end{figure}

  \textbf{Clause gadget.} 
  The construction of a clause gadget is shown in Fig.~\ref{fig:Clause}. We place a central vertex $c_1$ and we place three vertices $c_2$, $c_3$, and $c_4$ equidistant from it, connecting them to $c_1$. Next, we place three layers of three vertices each equidistant from $c_1$ such that each layer forms a triangle containing the central vertex. We place these such that all its vertices are placed well within the triangle $c_2c_3c_4$. We note that the only way to untangle this structure is to swap locations of the central vertex with one of $c_2$, $c_3$, and $c_4$ and orient the three layers in such a way that the edges missing in each layer line up towards the new location of $c_1$. This takes three swaps in total.
  
  \begin{figure}[bt]
\centering
  \includegraphics{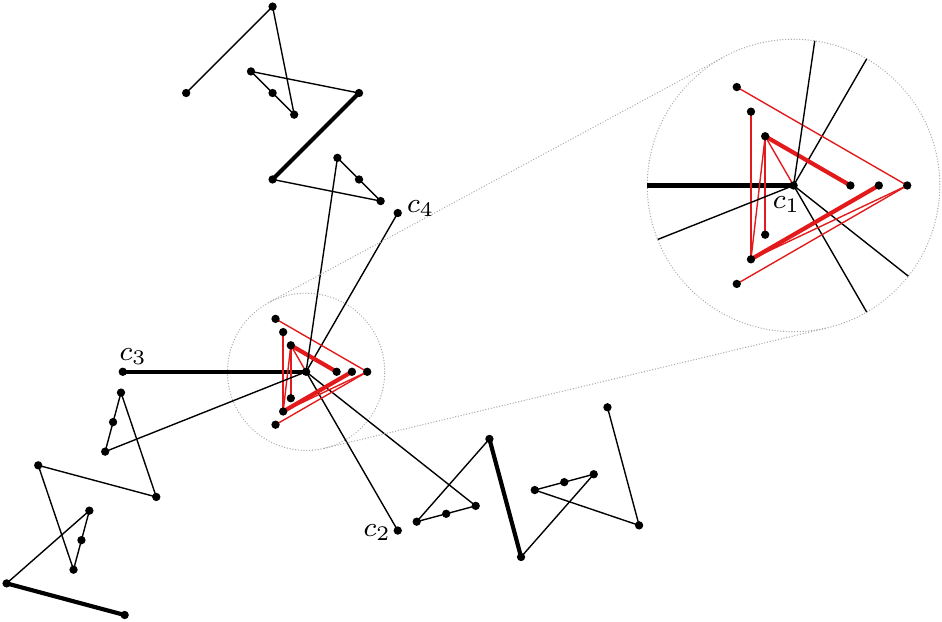}
  \caption{The clause gadget with enlargement of the central construction. Thick edges must be swapped to arrive at the drawing for Fig.~\ref{fig:SatisfiedClause}.}
  \label{fig:Clause}
\end{figure}

We connect a variable to a clause by using the split gadget at the variable, and then building a connection of an odd number of repetitions of the basic construction path, including one inherent in the split gadget to the clause, using one of $(c_1,c_2)$, $(c_1,c_3)$ or $(c_1,c_4$) for the last repetition of the basic construction. When placing this last repetition, we ensure that regardless of which of the three clause vertices $c_2$, $c_3$, and $c_4$ swaps with $c_1$, the edges connecting the variables to $c_1$ do not cross any of the other edges of the clause gadget. Furthermore, we ensure that the last crossing of the basic construction is untangled for free when $c_1$ is swapped with the vertex next to the clause (see Fig.~\ref{fig:Clause}).

In Fig.~\ref{fig:Clause} we illustrated two of these repetitions for each variable. This may either directly connect to the split gadget as shown in Fig.~\ref{fig:Split}, or use another even number of repetitions in between; we assume the former in our exposition here. But either case guarantees that the number of swaps required to untangle it is the same, regardless of whether the variable is \true or \false.

  When exactly one of the three variables has the value \true, we can untangle the clause gadget using five swaps. Note that we do not count the outermost swap of the \true variable, as this is shared with the split gadget. In Fig.~\ref{fig:Clause}, assume that the lower left variable is \true and the other two are \false. We untangle the clause gadget by swapping the edge connected to the center vertex that belongs to the \true variable (the horizontal edge in Fig.~\ref{fig:Clause}). We also swap two of the edges of the middle triangles to make that part plane. Finally, we swap the endpoints of the middle edge of the connecting gadget of the two \false variables. The result is shown in Fig.~\ref{fig:SatisfiedClause}. We note that the connection of the \true variable is untangled because the split gadget passes on its truth assignment.

\begin{figure}[tb] 
  \centering \includegraphics{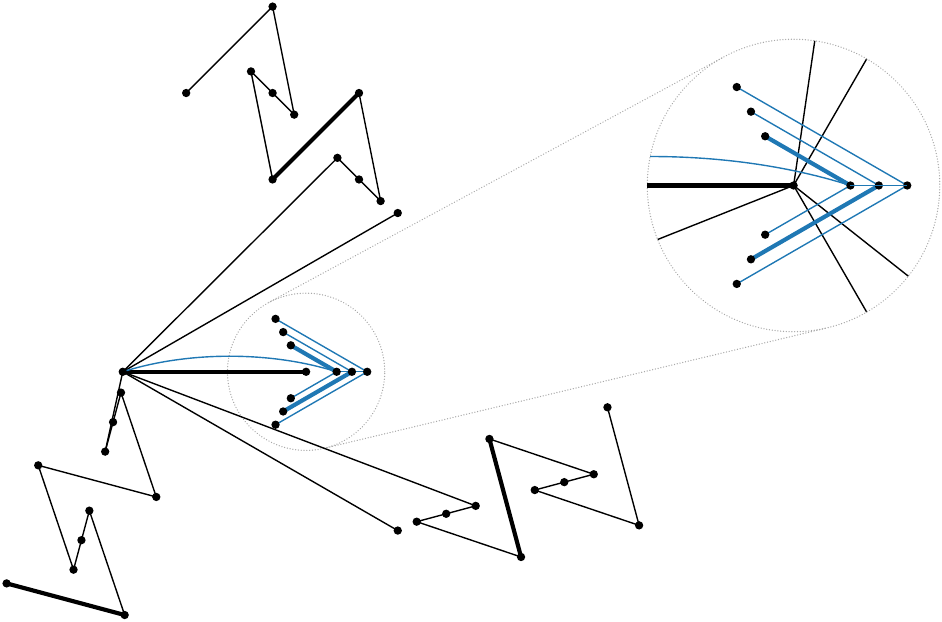}
  \caption{A clause satisfied by the bottom-left variable. Thick edges are those that have been swapped from Fig.~\ref{fig:Clause}.}
  \label{fig:SatisfiedClause}
\end{figure}

  Now consider the case where the clause is not satisfied. This can be because either no variable is \true or because at least two variables are \true. In the first case, the number of swaps needed to untangle the gadget is at least six, since we have to perform the same swaps as in the case where the clause is satisfied and in addition we need to swap the first edge of the connection of the third variable. This last edge swap corresponds to swapping the endpoints of the last edge in Fig.~\ref{fig:Split}c.

  Next, consider the case where the clause is not satisfied because at least two variables are \true. We first note that to untangle the clause, we need to swap one of the edges connected to the center as well as two edges of the triangles surrounding the center. One of the \true variables does not require any additional swaps, so let us consider the other two connections. For each of these connections that represents a \true variable, we note that we cannot swap the first edge of the connection, since it is fixed by the split gadget (see Fig.~\ref{fig:Split}b). Hence, to untangle an additional \true variable, we need to reverse its sequence of two horizontal middle edges; this reversal requires three swaps. If there are two \true variables, this implies that we need six swaps for the clause gadget. If there are three \true variables, seven swaps are needed.

  It remains to argue that there is no globally different set of swaps that makes the graph plane with fewer operations. Intuitively, for any pair of crossing edges, at least one of the neighboring edges needs to be swapped in order to remove the crossing. This argument implies that for edges that cross only a single other edge, we perform the minimum number of swaps to remove the crossings. Hence, we need to consider only those edges that cross multiple edges. Such edges occur only in the clause gadgets. We observe that the central vertex is surrounded by red triangles, hence to remove the crossings of the central vertex and these red edges, we need to either align the triangles (as our method does in the minimum number of operations) or swap at least one vertex of each triangle layer with a vertex outside the triangles. In order to perform the latter swaps, we would need to swap the endpoints of at least two edges per layer, since we first need to swap the central vertex with some edge outside the triangles followed by one or more operations to swap the desired triangle vertex with the central vertex. Hence, this approach requires at least six swaps in order to remove these crossings from a single clause gadget, which is more than the five swaps our approach needs. 
  
\begin{figure}[b]
    \centering
      \includegraphics[page=2]{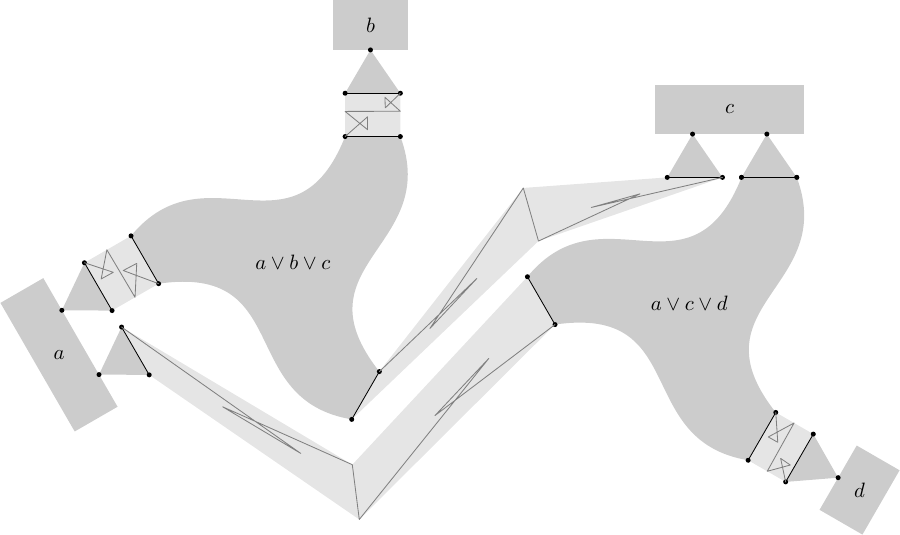}
    \caption{Schematic representation of a complete construction for $(a \vee b \vee c) \wedge (a \vee c \vee d)$. Each gadget is represented using a dark gray shape: a rectangle is a variable, a triangle is a split, and a spiral is a clause. Light gray areas represent a sequence of basic constructions to connect splits at variables to clauses.}
    \label{fig:fullconstruction}
  \end{figure}
  
  Hence, when a clause is not satisfied, untangling it takes more than five swaps. Finally, since the number of swaps required for the variable gadgets and split gadgets is the same regardless of whether the instance is satisfiable or not, an instance is satisfiable if and only if we need five swaps per clause to untangle all clauses.

  \textbf{Constructing a tree.} 
  Each of the gadgets above can be translated, scaled and rotated. We place the gadgets to adhere to the given embedding of the positive planar 1-in-3-SAT instance, using basic constructions to connect clause gadgets with their split gadgets at each relevant variable gadget. This is schematically shown in Fig.~\ref{fig:fullconstruction}. However, this construction is generally not a tree, as it can contain cycles.
  To construct a tree, we remove the middle edges from some basic constructions (see Fig.~\ref{fig:MakingTree}). Since the endpoints of these edges are never swapped in any satisfiable assignment, this does not influence the satisfiability of the instance.

  \begin{figure}[htb]
    \centering
      \includegraphics{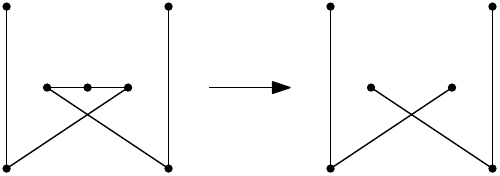}
    \caption{Removing edges from the basic construction to turn the graph into a tree.}
    \label{fig:MakingTree}
  \end{figure}

  This last step shows that we can solve an instance of positive planar 1-in-3-SAT by constructing a tree and determining whether the clause gadgets can be untangled using five swaps per clause. Retrieving the variable assignment for positive planar 1-in-3-SAT can be done by checking how the corresponding variable gadgets are untangled. Hence, the problem is NP-complete.
\end{proof}

\section{Generating levels}

In this section we describe how puzzle instances or levels can be generated for \textsc{Swap Planarity}. 
First we outline a five-step procedure, and then we explain these steps
in more detail. We pay attention to three properties: (i) the puzzle instance should
look good, also in states to be reached later, (ii) every possible good puzzle
instance should be a possible output, for diversity, and (iii) solutions
should not have a particular structure that might be identified by a puzzler,
which may upset the intended puzzle instance difficulty.

\subsection{Process of level generation}

We describe a five-step procedure to generate a puzzle instance.
We assume that a desired number $n$ of vertices is specified, and also
a desired number $m$ of edges, and a desired minimal number $s$ of swaps to the
solution.

\begin{enumerate}
\item
Generate a set $V$ of $n$ points in a playing area, such that for no
two points, an edge between them would visually conflict with any other point from $V$ (property (i)).
\item
\label{step:Delaunay}
Generate a Delaunay triangulation on $V$, leading to an edge set $E''$.
\item
\label{step:flip}
Perform a number of Lawson flips to make sure that the solution of the
puzzle instance need not only have Delaunay edges (done for properties (ii) and (iii)). 
This makes $E'$ out of~$E''$.
\item
\label{step:remove}
Remove a number of edges at random from $E'$ until $m$ edges remain.
Make sure that no isolated vertices remain. This gives the edge set $E$.
\item
Perform $s$ swap operations at random, by picking edges at random from $E$.
Test if the resulting instance requires $s$ swaps to a planar state (and if
not, swap more edges).
\end{enumerate}
The whole process ensures property (ii): any puzzle instance that satisfies property (i) can be generated, provided that sufficiently many flips are performed in step 3.

\subsection{Generating points}

Given the shape of screens, it is natural to generate a point set in a square or rectangular
region. There are two important issues to consider when generating point sets. First, collinearity or near-collinearity of points means that potentially, an edge will partly overlap a vertex in the drawing. 
This is undesirable. Second, point sets are ``combinatorially different'', which relates to the variation to be obtained in puzzle instances. We discuss
these two issues next.

Let us assume that each vertex is drawn as a disc with radius $\rho$. Then any two vertices (centers) should be separated more than $2\rho$ in order for their
discs to be disjoint.
Every edge is drawn as a rectangle with its length matching the distance between its endpoints ($>2\rho$) and width $\lambda<2\rho$.
The center of each vertex should be further away than $\rho+\lambda/2$ from
the center line of any edge that it is not incident to~\cite{van2011bold}.
To have a little more room around each vertex and edge we introduce a parameter $\delta$ that specifies for each point how far it must
be from each other point and edge, when points are viewed as $0$-dimensional and edges as $1$-dimensional. We always choose
$\delta > 2\rho$.

\begin{definition}
Given $\delta>0$, a set $P$ of points in the plane is in \emph{$\delta$-general position} if and
only if for any three distinct points $p,q,r\in P$, the distance from $r$ to the line through $p$ and $q$ is at least $\delta$.
\end{definition}

To generate a point set in $\delta$-general position, we incrementally add points,
uniformly distributed in a square. 
For each addition, we check if the $\delta$-general position condition is violated, and if so, we discard the last added point. 
To test this condition, we consider every pair of accepted points with the newly added point.
Using a bit of geometry we can identify a region bounded by six lines where the
new point may not lie, see Fig.~\ref{fig:collinear}. 
Two of these lines are the outer tangents to two discs
of radius $\delta$ centered on the two accepted points. 
The other four are tangents to one of these discs, passing through the other accepted point.
Hence, this test can be done in quadratic time per new point.

Note that a set of two points is always in $\delta$-general position, but any third point enforces all points in the set to be at least distance $\delta > 2\rho$ apart; hence, we do not need to check vertex-vertex distances after adding the second point.

  \begin{figure}[b]
    \centering
      \includegraphics{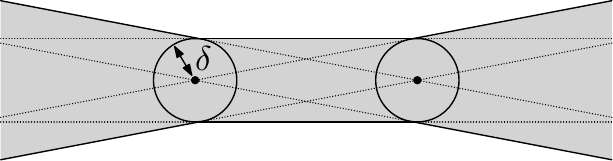}
    \caption{Region where a third point may not be placed if $\delta$-general position should be preserved.}
    \label{fig:collinear}
  \end{figure}

When we generate puzzles with a considerable number of points we may get many failures.
It is possible to compute the whole region where new points can be placed by
generating the quadratically many regions for the accepted points and computing their union.
The complement of this union is where a new point can still lie. In particular, we can
compute this union and sample the complement explicitly, which means we do not get failures.
If the union covers the whole square, we cannot add points anymore. 
For $a$ accepted points so far, this union has complexity $O(a^4)$ and 
can be computed in $O(a^4)$ time~\cite{halperin2018arrangements}.

We next discuss the issue of combinatorially different point sets. To understand what this
means, imagine a set of $n$ points in convex position: they all lie on the convex hull.
Whether points lie as the vertices of a regular $n$-gon, or spread on an ellipse, or more
randomly placed (but still in convex position), these point sets are essentially the same
from the perspective of intersecting edges between these points. 
Any graph on these points has the same intersecting edges regardless of
where the points lie precisely. 
Moreover, any point set with $n$ points and $k$ on its convex hull ($3 \leq k \leq n$) has at most $3n-k-3$ edges that do
not intersect (the fewer vertices on the convex hull, the more edges can be in a plane graph).
Point sets with the same number of points but different numbers of points on the convex hull are combinatorially different.
But there are still differences between point sets with the same numbers of points and the same number of points on the convex hull. In Section~\ref{sec:instanceequivalence} we look more closely at the concept of puzzle equivalence.
In our experiments we use the number of points inside the convex hull ($n-k$) as a simple indicator of how varied instances may be. If there are no points inside, then the point sets are effectively the same; more points inside allow for more combinatorial differences.

\subsection{Generating a plane graph}

Once we have generated a set $V$ of $n$ points without collinearity or closeness, we can
generate edges. We generate a plane graph (a solution) to a puzzle instance in three 
steps (steps~\ref{step:Delaunay}--\ref{step:remove}).

First, we compute the Delaunay triangulation of $V$~\cite{bcko-cgaa-08}. This is a specific triangulation
of a point set that maximizes the smallest angle that is used in the triangulation.
This triangulation is also characterized by the empty-circle property: for any two points $v_i$ and $v_j$
for which a circle exists that touches only $v_i$ and $v_j$ and which has no points of
$V$ inside, there is an edge connecting $v_i$ and $v_j$. This characterization
(in general) completely specifies the triangulation. There are several known
algorithms to compute the Delaunay triangulation of $n$ points in $O(n \log n)$ time.
This gives the edge set $E''$.

Second, we perform a few Lawson flips (beware that flips and swaps are very different operations). A Lawson flip can be applied to a pair of
edge-adjacent triangles in a triangulation if those triangles together form a
\emph{convex} quadrilateral. A Lawson flip removes the shared edge and re-triangulates
the resulting quadrilateral in the (only) other way. These 
flips make it harder for a puzzler to solve instances. Delaunay triangulations
favor shorter edges, and Lawson flips can generate longer edges again. If a puzzler
would know---or realize---that the solution to each puzzle instance uses only
Delaunay edges, then (s)he can quickly see which edges must be avoided in the drawing
by imagining the empty-circle test (let's face it: these puzzles are going to be
done by geometers). Edges to be flipped are selected randomly, and the flip is done
only if the four involved vertices are in convex position (otherwise the resulting
drawing would be non-planar). The resulting edge set is denoted $E'$.

Third, we remove some edges from $E'$ so that a puzzle instance solution is not
always a triangulation. We ensure that no isolated vertices are created, by not removing edges with an endpoint of degree 1.
These would not influence the puzzle or its solution in any way. Notice that an
isolated edge does influence the puzzle. While a swap applied to such an edge
does not change the drawing, swapping other edges may resolve edge intersections with
the isolated edge. 

By removing edges we can realize a desired number of edges in the solution.
Removing many edges may cause the puzzle instance to have multiple solutions 
and become easy.

\subsection{Generating an instance}

We have now generated a graph with a specified number of vertices and edges, 
and in particular, a solution to this puzzle instance. To generate the puzzle
instance itself we make some swaps such that undoing these (swapping the same
edges in reverse order) solves the instance.

It appears that puzzle instances with just two or three swaps from a solution are
already not so easy (Fig.~\ref{fig:introex}). Once a player gets more experienced, instances with four
swaps may become suitable. This means that testing the difficulty of a solution
can be done by brute-force. For example, a graph with $20$ edges that should be
four swaps away from a solution can be tested by trying all $20\cdot 19^3=137,180$
possibilities (we exclude swapping the same edge twice in a row). 
This may lead to an instance with fewer necessary swaps to solve than we have used to generate it; in this case we perform extra swaps until the desired minimal number of swaps is obtained.
We will also recognize if there
are more ways to a solved state, making the instance a bit easier too. Finally,
swaps that are independent and possibly even well-separated also give
rise to easier instances. Two swaps are independent if the four endpoints of the edges
are disjoint and there is no other edge than the two that are swapped between these four
vertices.

We have now realized the three properties we aimed for.
The visual quality (i) of the instance and every intermediate state that can be reached
is captured by the vertex-vertex distance and vertex-edge distance conditions.
The puzzle diversity (ii) is realized by allowing any number of vertices, edges, and steps
to the solution, every possible plane drawing as a solution, and every possible non-plane
drawing as a puzzle instance. There is no puzzle instance that cannot be generated.
Absence of unintended structure (iii) is accomplished by ensuring that for a point set, any
edge between two points could be part of the solution.

\section{Implementation and experiments}

The {\sc Swap Planarity} game is implemented using Unity. Besides trying the game to see how difficult and fun puzzle instances are, we are interested in the efficient generation of non-collinear point sets, the number of points on the convex hull, the non-collinearity parameter $\delta$, and relations these.

\begin{figure*}[htb]
\centering
\includegraphics[scale=0.19]{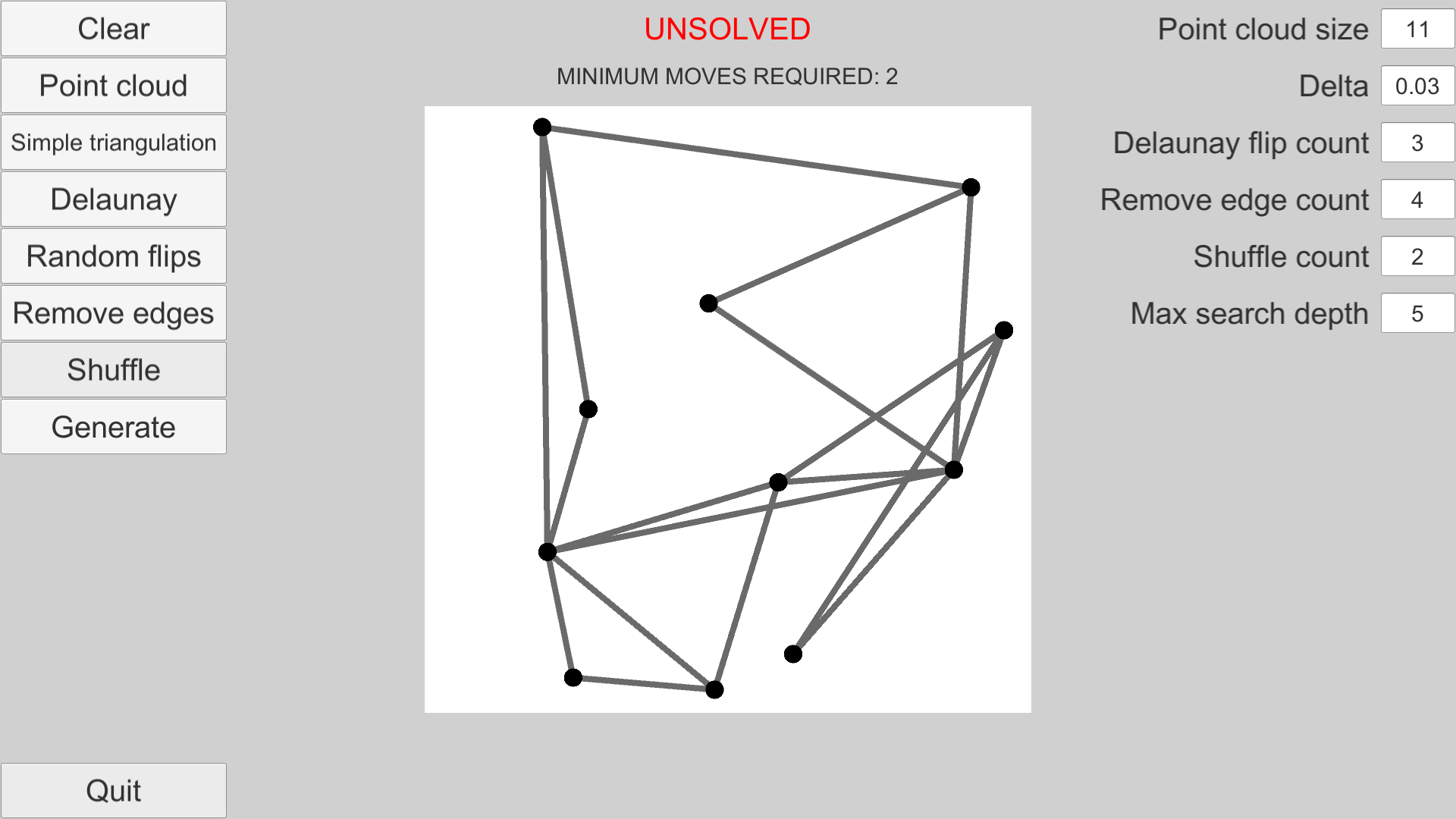}~~
\includegraphics[scale=0.25]{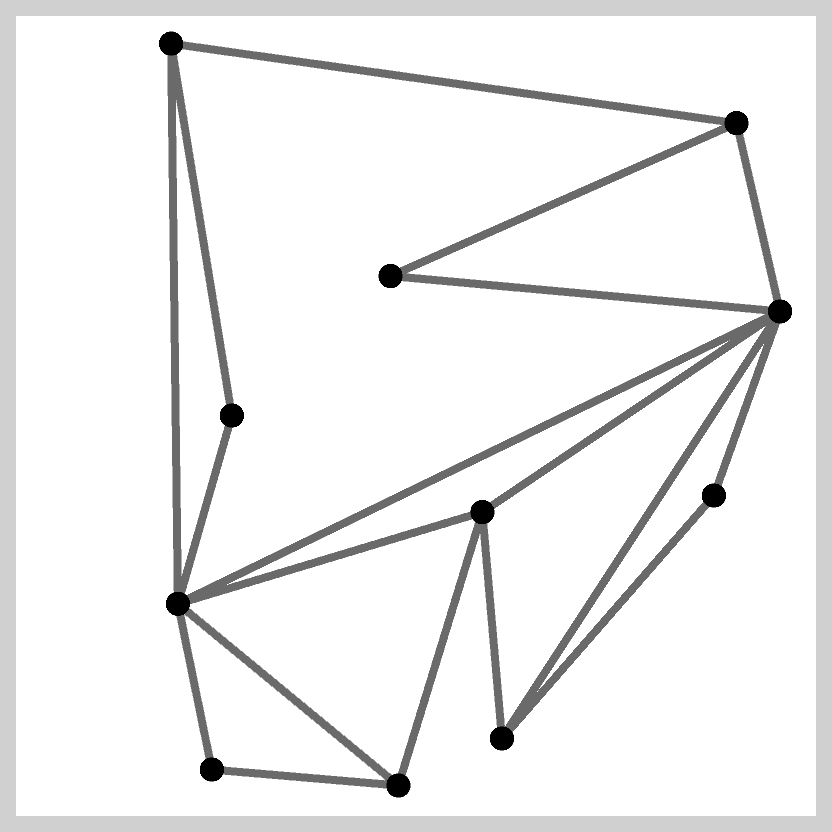}
\caption{Left figure, screenshot with the steps of the generation listed in sequence (Generate does all steps in order) and the settings used. Right figure, the solution of this puzzle instance.}
\label{fig:screenshot}
\end{figure*}

Fig.~\ref{fig:screenshot} shows the interface. From the settings on the right we can see that the instance has $11$ points generated with $\delta=0.03$ to ensure non-collinearity, the initial triangulation is $3$ flips away from being Delaunay, then $4$ edges were removed and two swaps were performed to shuffle the planar graph. The solution is shown on the right.

When we try to generate a large point set with a large value of $\delta$, we may fail because there may not be enough space on the screen (play area) to realize the separation. This also depends on the random generation itself. It can happen that a point set of $14$ points cannot be extended to $15$ points without violating collinearity, but sets of $15$ non-collinear points may
still exist. This means that the point generation procedure may have to abort and restart. If aborting is done too early, generation may be inefficient because we start from scratch without having to. If aborting is done too late, generation may have spent a lot of time on a configuration that cannot be extended anymore. Fig.~\ref{fig:graph-performance} illustrates this for a fixed value of~$\delta$; data points we generated with intervals of $50$ between $0$ and $500$ and with intervals of $500$ after that. Note that the vertical axis has exponential scale. For the larger point set sizes we observe that we should make enough attempts to add a point, but not too many, to get the best efficiency.

\begin{figure*}[tb]
\centering
\includegraphics[width=0.72\textwidth]{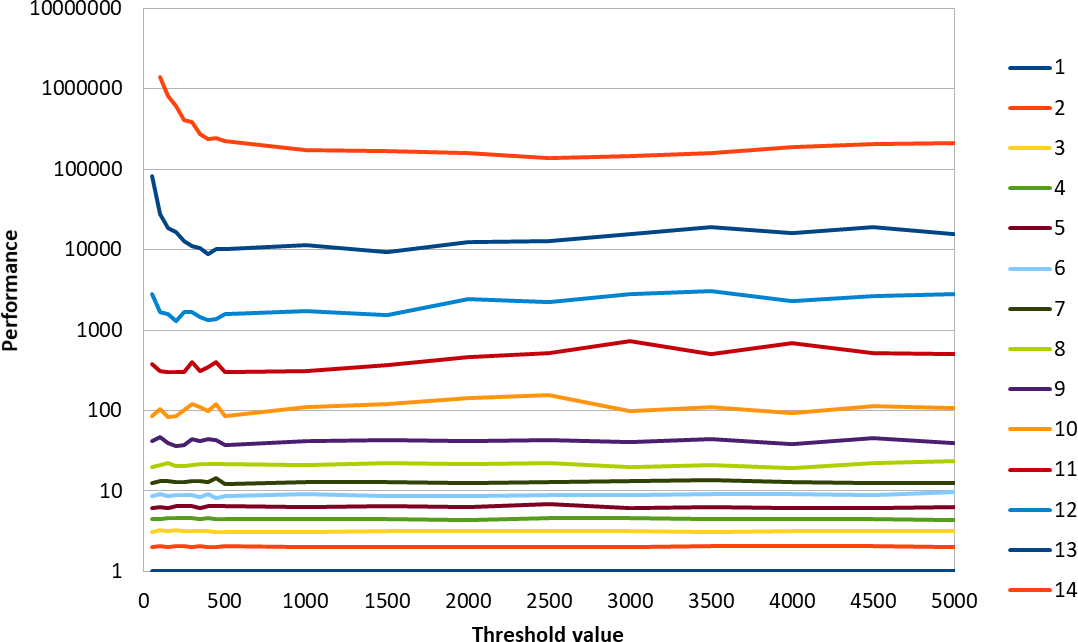}
\caption{Performance (total number of attempts to add a point to generate a complete point set) as a function of threshold choice, for different point set sizes. The threshold value represents the total number of attempts to add a random point before the generation of a point set is aborted and restarted.}
\label{fig:graph-performance}
\end{figure*}
\begin{figure*}
\centering
\includegraphics[width=0.60\textwidth]{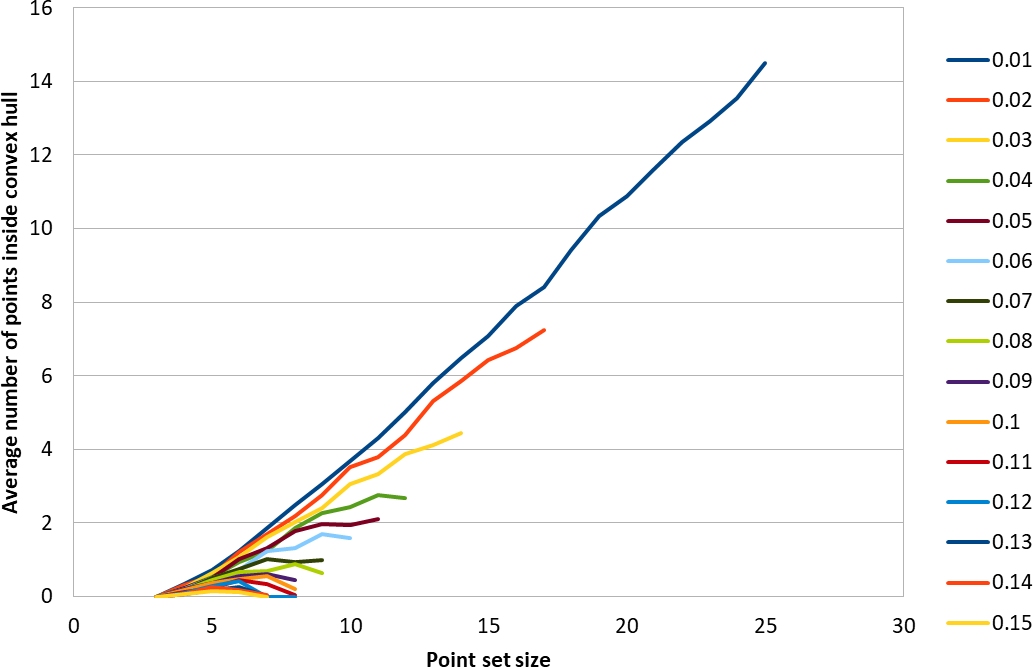}\hfill
\includegraphics[width=0.37\textwidth]{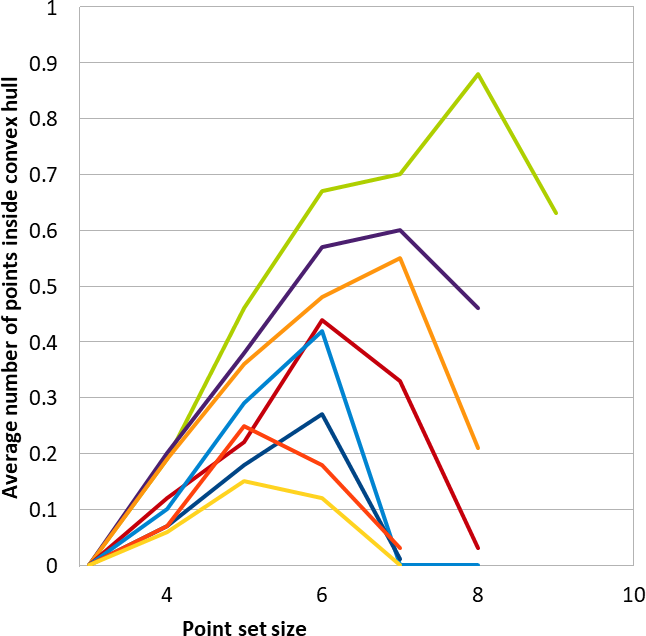}
\caption{Number of points inside the convex hull as a
function of the point set size, for different thresholds. Data points are averaged over $100$ point set instances that were generated. To the right, detail of the same figure.}
\label{fig:graph-ch1}
\end{figure*}

We also determined the number of points inside the convex hull for different point set sizes and different values of~$\delta$. We noticed a surprising phenomenon: the larger $\delta$, the fewer points are in the convex hull. 
This can be seen in Fig.~\ref{fig:graph-ch1}, right: for increasing $\delta$, fewer points tend to lie inside the convex hull. This happens especially when it gets difficult to generate larger point sets for a given $\delta$, and hence we cannot observe the behavior for larger point set values in Fig.~\ref{fig:graph-ch1}, left. It may be the case that a placement of points on the convex hull is a good placement if one wants to realize a large $\delta$. This suggestion is supported by theory on bold graph drawings~\cite{pach2011every}.
Fig.~\ref{fig:graph-sd1} shows the standard deviations over the $100$ point set instances. It also shows that if $\delta$ is chosen relatively large, fewer points will be inside the convex hull.

\begin{figure*}[tb]
\centering
\includegraphics[width=0.72\textwidth]{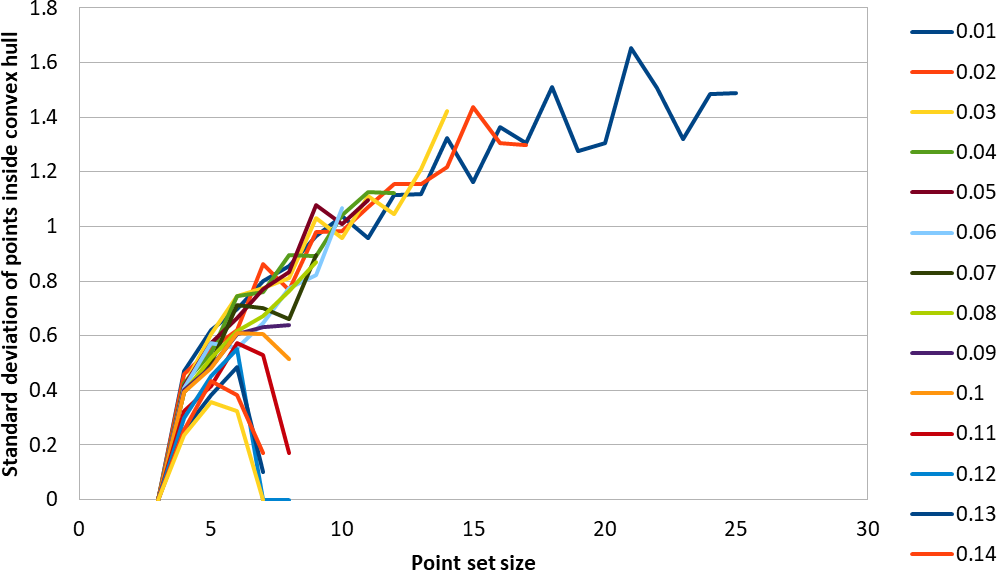}
\caption{Standard deviation of the number of points inside the convex hull as a function of the point set size, for different thresholds.}
\label{fig:graph-sd1}
\end{figure*}

The experiments show the following trade-off: puzzle instances with a good visual appearance (clear non-collinearity, large $\delta$) are harder to generate efficiently and show less diversity, indicated by the relatively large number of points on the convex hull.

\section{Equivalence of instances}
\label{sec:instanceequivalence}

We now return our attention to puzzle diversity. 
To generate diverse sets of puzzle instances, we ideally need a good measure of puzzle instance similarity (or its converse, distance).
We study this aspect in its weakest form, namely puzzle instance equivalence. 
The following definition essentially states that two puzzle instances should be
considered swap-equivalent if and only if any sequence of
corresponding swaps in both drawings gives the same
sets of intersecting edges in the drawings.
With slight abuse of notation we use the same symbol
for a vertex in a graph and the point in the plane
where it is drawn.

To define swap-equivalence we use a one-to-one matching $\mu$ between the vertices of the two
graphs and their drawings.
We use the subscripts $1$ and $2$ to refer to the two graphs, their drawings and their vertices. Furthermore, we use the same letter for graph vertices or points that are matched. So, vertex $u_1$ in graph $G_1$ is matched with vertex $u_2$ in~$G_2$. We use $\mu$ as an invertible function $\mu \colon V_1 \rightarrow V_2$. Thus, generally, we use $\mu(u_1) = u_2$ and $\mu^{-1}(u_2) = u_1$.

\begin{definition}
\label{def:swap-equivalent}
Drawings $D_1$ and $D_2$ of graphs $G_1$ and $G_2$
are \emph{swap-equivalent} if and only if there is
a one-to-one matching $\mu$ between their vertices such that:
\begin{description}
    \item[(i)]
    $(u_1,v_1)$ is an edge in $G_1$ if and only if
    $(u_2,v_2)$ is an edge in $G_2$, where $\mu(u_1) = u_2$ and $\mu(v_1) = v_2$;
    \item[(ii)]
    after any sequence of zero or more swaps of
    matched edges in both graphs,
    $(u_1,v_1)$ and $(w_1,x_1)$ intersect
    if and only if matched edges $(\mu(u_1),\mu(v_1))$ and $(\mu(w_1),\mu(x_1))$
    intersect.
\end{description}
\end{definition}

It is clear that the drawings $D_1$ and $D_2$ need to
have the same number of vertices and edges, otherwise
one-to-one matchings cannot exist.

\begin{lemma} 
\label{lem:iso}
If two drawings $D_1$ and $D_2$ of graphs $G_1$ and $G_2$ are swap-equivalent, then
their graphs are isomorphic.
\end{lemma}
\begin{proof}
This follows directly from Definition~\ref{def:swap-equivalent}(i).
\end{proof}

Beyond the topological equivalence of the graph indicated by the lemma above, we also need a form of geometric equivalence for the point set of the two puzzle instances. As it turns out, this corresponds to the order type, as formalized in the lemma below.
Consider two point sets $P_1$ and $P_2$ of $n$ points each. These point sets are combinatorially equivalent if a one-to-one mapping $\mu \colon P_1 \rightarrow P_2$ exists such that for any three points $u_1,v_1,w_1 \in P_1$, the sequence $u_1v_1w_1$ is a left
turn if and only if the sequence $\mu(u_1)\mu(v_1)\mu(w_1)$ of points from $P_2$ is a left turn.
The equivalence class thus obtained is called an \emph{order type}~\cite{aichholzer2002enumerating,aichholzer2007number}. 

\begin{lemma}
\label{lem:ot}
If two drawings $D_1$ and $D_2$ of connected, non-star graphs $G_1$ and $G_2$ are swap-equivalent, then
their point sets have the same order type.
\end{lemma}
\begin{proof}
Let $P_1$ and $P_2$ denote the point sets of the two graphs.
Suppose these point sets have a different order type. Then
there is a quadruple of points $u_1,v_1,w_1,x_1$ in $P_1$
with matching points $u_2 = \mu(u_1), v_2 = \mu(v_1), w_2 = \mu(w_1), x_2 = \mu(x_1)$ in $P_2$ such
that $u_1,v_1,w_1,x_1$ are in convex position and
$u_2,v_2,w_2,x_2$ are not. Assume without loss of
generality that $(u_1,v_1)$ intersects $(w_1,x_1)$.
Since $G_1$ is connected, we can realize any mapping
of graph vertices to points by Lemma~\ref{lem:upperbound}.
Since $G_1$ is connected and not a star graph, it has two edges with four distinct vertices. Consider a sequence of swaps that places these two edges on $(u_1,v_1)$ and $(w_1,x_1)$.
Then these edges intersect. However, since $u_2,v_2,w_2,x_2$ are not in convex position in $D_2$, the corresponding edges in $D_2$ do not intersect.
This contradicts Definition~\ref{def:swap-equivalent}(ii).
\end{proof}

\begin{theorem}
Two drawings $D_1$ and $D_2$ of connected, non-star
graphs $G_1$ and $G_2$ using point sets $P_1$ and $P_2$
are \emph{swap-equivalent} if and only if there is
a one-to-one matching between $P_1$ and $P_2$ such that:
\begin{itemize}
    \item
    point sets $P_1$ and $P_2$ have the same order type which respects the one-to-one matching;
    \item
    the graphs $G_1$ and $G_2$ are isomorphic with
    the given one-to-one matching of the points in $P_1$ and $P_2$ applied to the corresponding vertices in $G_1$ and $G_2$.
\end{itemize}
\end{theorem}
\begin{proof}
Assume two drawings $D_1$ and $D_2$ are swap-equivalent.
Then Lemma~\ref{lem:iso} shows that their graphs are isomorphic.
Furthermore, Lemma~\ref{lem:ot} shows that their point
sets must have the same order type. 
Now assume that the order types are the same and the graphs are isomorphic, but there is no one-to-one matching that simultaneously witnesses the same order type and graph isomorphy.
Consider any one-to-one matching for graph isomorphy. Then there is a quadruple of points that are in convex position in $G_1$ whose matched points are not in convex position in $G_2$, or vice versa. The same argument as the one
used to prove Lemma~\ref{lem:ot} shows that the intersections of edges are not the same in the two drawings after some sequence of swaps, contradicting swap-equivalence.

Next, assume that two drawings satisfy the two conditions of the theorem. 
Then there is a one-to-one matching between
the vertices that respects the order types and that
is a witness for graph isomorphism at the same time.
Graph isomorphism implies (i) of Definition~\ref{def:swap-equivalent}, and having the same
order type implies that the same pairs of edges intersect
in the complete graph. Hence this is also true in any subgraph.
\end{proof}

The theorem above implies an efficient way to
test whether two puzzle instances are equivalent.
We first identify the at most $n$ one-to-one matchings
for the order type, and then check whether this
matching also realizes graph isomorphism.
Generating and testing the up to $n$ matchings takes $O(n^3)$ time~\cite{goodman1984semispaces}, and testing isomorphism for a given
matching takes time linear in the size of the graph.
Hence, swap-equivalence of two drawings with $n$ vertices can be tested in $O(n^3)$ time.

\section{Conclusions}

We introduced a new graph planarity puzzle game called {\sc Swap Planarity} and analyzed various properties, including the algorithmic complexity of solving instances.
Any instance that can be solved, is solved in $O(n^2)$ swaps. However, deciding if an instance can be solved is NP-complete. When the graph is a tree, the instance can always be solved, but deciding if $k$ swaps are sufficient is again NP-complete.
We presented a method to generate instances effectively while paying attention to visual clarity, diversity, and absence of accidental structure. Our implementation shows that generation works well, but has a trade-off between a good visual clarity on the one hand and diversity and efficient generation on the other.

Visual clarity was defined using a new, simple condition on point sets called $\delta$-general position. The experiments showed an interesting phenomenon, namely that if $\delta$ is fairly large for the available space and the number of points, the solution tends to be a set of points in convex position (and no points in the interior of the convex hull).
It would be interesting to explore this relationship further.

We think that the new, swap-based graph planarity puzzle game is a nice, elegant addition to the collection of abstract puzzle games. The puzzle is NP-hard, the number of crossings may need to be increased to reach a solution, and even small instances are not so easy to solve. User studies are needed to analyze the fun and difficulty of the game for players.

\section*{Acknowledgments} 
Part of this work was performed at the Dutch-Japanese Bilateral Seminar on Kinetic Geometric Networks. We thank the organizers and participants of the workshop for providing a fun and stimulating research environment. A.v.R. was supported by JST ERATO Grant Number JPMJER1201, Japan. M.v.K. was supported by the Netherlands Organisation for Scientific Research on grant no.~612.001.651. W.M. was supported by Netherlands eScience Center (NLeSC) on grant no. 027.015.G02.

\bibliographystyle{plain}
\bibliography{references}

\begin{thebibliography}{10}

\bibitem{aichholzer2002enumerating}
Oswin Aichholzer, Franz Aurenhammer, and Hannes Krasser.
\newblock Enumerating order types for small point sets with applications.
\newblock {\em Order}, 19(3):265--281, 2002.

\bibitem{aichholzer2007number}
Oswin Aichholzer, Thomas Hackl, Clemens Huemer, Ferran Hurtado, Hannes Krasser,
  and Birgit Vogtenhuber.
\newblock On the number of plane geometric graphs.
\newblock {\em Graphs and Combinatorics}, 23(1):67--84, 2007.

\bibitem{bose2002embedding}
Prosenjit Bose.
\newblock On embedding an outer-planar graph in a point set.
\newblock {\em Computational Geometry}, 23(3):303--312, 2002.

\bibitem{BoseDHLMW09}
Prosenjit Bose, Vida Dujmovic, Ferran Hurtado, Stefan Langerman, Pat Morin, and
  David~R. Wood.
\newblock A polynomial bound for untangling geometric planar graphs.
\newblock {\em Discrete {\&} Computational Geometry}, 42(4):570--585, 2009.

\bibitem{cabello2006planar}
Sergio Cabello.
\newblock Planar embeddability of the vertices of a graph using a fixed point
  set is {NP}-hard.
\newblock {\em Journal of Graph Algorithms and Applications}, 10(2):353--363,
  2006.

\bibitem{chiba1985linear}
Norishige Chiba, Takao Nishizeki, Shigenobu Abe, and Takao Ozawa.
\newblock A linear algorithm for embedding planar graphs using {PQ}-trees.
\newblock {\em Journal of Computer and System Sciences}, 30(1):54--76, 1985.

\bibitem{bcko-cgaa-08}
Mark de~Berg, Otfried Cheong, Marc van Kreveld, and Mark Overmars.
\newblock {\em Computational Geometry -- Algorithms and Applications}.
\newblock Springer, Berlin, 3rd edition, 2008.

\bibitem{de1990draw}
Hubert De~Fraysseix, J{\'a}nos Pach, and Richard Pollack.
\newblock How to draw a planar graph on a grid.
\newblock {\em Combinatorica}, 10(1):41--51, 1990.

\bibitem{goaoc2009untangling}
Xavier Goaoc, Jan Kratochv{\'\i}l, Yoshio Okamoto, Chan-Su Shin, Andreas
  Spillner, and Alexander Wolff.
\newblock Untangling a planar graph.
\newblock {\em Discrete \& Computational Geometry}, 42(4):542--569, 2009.

\bibitem{goodman1984semispaces}
Jacob~E. Goodman and Richard Pollack.
\newblock Semispaces of configurations, cell complexes of arrangements.
\newblock {\em Journal of Combinatorial Theory, Series A}, 37(3):257--293,
  1984.

\bibitem{halperin2018arrangements}
Dan Halperin and Micha Sharir.
\newblock Arrangements.
\newblock In Jacob~E. Goodman, Joseph O'Rourke, and Csaba~D. T\'oth, editors,
  {\em Handbook of Discrete and Computational Geometry}, pages 723--762. 3rd
  edition, 2018.

\bibitem{mulzer2008minimum}
Wolfgang Mulzer and G{\"u}nter Rote.
\newblock Minimum-weight triangulation is {NP}-hard.
\newblock {\em Journal of the ACM}, 55(2):11:1--11:29, 2008.

\bibitem{pach2011every}
J{\'a}nos Pach.
\newblock Every graph admits an unambiguous bold drawing.
\newblock In {\em International Symposium on Graph Drawing}, pages 332--342,
  2011.

\bibitem{pach1991embedding}
J{\'a}nos Pach, Peter Gritzmann, Bojan Mohar, and Richard Pollack.
\newblock Embedding a planar triangulation with vertices at specified points.
\newblock {\em The American Mathematical Monthly}, 98(2):165--166, 1991.

\bibitem{tamassia2013handbook}
Roberto Tamassia.
\newblock {\em Handbook of Graph Drawing and Visualization}.
\newblock CRC Press, 2013.

\bibitem{Planarity}
John Tantalo.
\newblock Planarity.
\newblock http://planarity.net/, 2007.
\newblock Accessed: 2018-05-25.

\bibitem{tisue2004netlogo}
Seth Tisue and Uri Wilensky.
\newblock {NetLogo}: A simple environment for modeling complexity.
\newblock In {\em International Conference on Complex Systems}, volume~21,
  pages 16--21, 2004.

\bibitem{van2011bold}
Marc van Kreveld.
\newblock Bold graph drawings.
\newblock {\em Computational Geometry}, 44(9):499--506, 2011.

\bibitem{Verbitsky08}
Oleg Verbitsky.
\newblock On the obfuscation complexity of planar graphs.
\newblock {\em Theoretical Compututer Science}, 396(1-3):294--300, 2008.

\bibitem{YamanakaDIKKOSS15}
Katsuhisa Yamanaka, Erik~D. Demaine, Takehiro Ito, Jun Kawahara, Masashi
  Kiyomi, Yoshio Okamoto, Toshiki Saitoh, Akira Suzuki, Kei Uchizawa, and
  Takeaki Uno.
\newblock Swapping labeled tokens on graphs.
\newblock {\em Theoretical Computer Science}, 586:81--94, 2015.

\end{thebibliography}

\end{document}